\documentclass[11pt]{article}

\newcommand{\prl}{Phys. Rev. Lett.~}
\newcommand{\pra}{Phys. Rev. A~}

\usepackage{color}
\usepackage{subfigure}
\usepackage[T1]{fontenc}
\usepackage[sc]{mathpazo}
\usepackage{amsmath}
\usepackage{amssymb}
\usepackage{graphicx,color}
\usepackage{framed}
\usepackage{multirow}
\usepackage{enumerate}
\usepackage{amsthm}
\usepackage{amsfonts,mathrsfs}
\usepackage{geometry} 
\usepackage{eepic}
\usepackage{ifthen}
\usepackage[vcentermath]{youngtab}
\usepackage[unicode=true,pdfusetitle, bookmarks=true,bookmarksnumbered=false,bookmarksopen=false, breaklinks=false,pdfborder={0 0 0},backref=false,colorlinks=false] {hyperref}
\hypersetup{
colorlinks,linkcolor=myurlcolor,citecolor=myurlcolor,urlcolor=myurlcolor}
\definecolor{myurlcolor}{rgb}{0,0,0.7}

\usepackage{color}

\newcommand{\blue}{\textcolor{blue}}

\usepackage{hyperref}
\hypersetup{pdfpagemode=UseNone}

\newcommand{\tinyspace}{\mspace{1mu}}

\newcommand{\abs}[1]{\left\lvert\tinyspace #1 \tinyspace\right\rvert}

\renewcommand{\det}{\operatorname{det}}
\renewcommand{\t}{{\scriptscriptstyle\mathsf{T}}}

\newcommand{\setft}[1]{\mathrm{#1}}

\newcommand{\density}[1]{\setft{D}\left(#1\right)}

\def \dif {\mathrm{d}}
\def \diag {\mathrm{diag}}

\def\complex{\mathbb{C}}
\def\real{\mathbb{R}}

\newenvironment{mylist}[1]{\begin{list}{}{
    \setlength{\leftmargin}{#1}
    \setlength{\rightmargin}{0mm}
    \setlength{\labelsep}{2mm}
    \setlength{\labelwidth}{8mm}
    \setlength{\itemsep}{0mm}}}
    {\end{list}}

\def\ot{\otimes}

\newcommand{\out}[2]{| #1\rangle\langle #2 |}
\newcommand{\Inner}[2]{\left\langle #1 , #2\right\rangle}

\newcommand{\Pa}[1]{\left(#1\right)}

\newcommand{\Br}[1]{\left[#1\right]}
\newcommand{\set}[1]{\{#1\}}
\newcommand{\Set}[1]{\left\{#1\right\}}

\newcommand{\ket}[1]{|#1\rangle}

\DeclareMathOperator{\trace}{Tr}

\newcommand{\Ptr}[2]{\trace_{#1}\Pa{#2}}

\newcommand{\Tr}[1]{\Ptr{}{#1}}

\def\cE{\mathcal{E}}

\def\cO{\mathcal{O}}

\def\bD{\mathbf{D}}

\def\rF{\mathrm{F}}\def\rH{\mathrm{H}}

\def\rS{\mathrm{S}}
\def\rU{\mathrm{U}}

\def\L{\textsf{L}}

\newtheorem{thrm}{Theorem}[section]

\newtheorem{prop}[thrm]{Proposition}
\newtheorem{cor}[thrm]{Corollary}
\theoremstyle{definition}
\newtheorem{definition}[thrm]{Definition}

\numberwithin{equation}{section}

\newcounter{questionnumber}

\begin{document}

\title{\Large \bf Spectral density of mixtures of random density matrices for qubits}

\author{\blue{Lin Zhang}$^\text{a}$\footnote{E-mail: godyalin@163.com; linyz@hdu.edu.cn},\quad\blue{Jiamei Wang}$^\text{b}$,\quad\blue{Zhihua Chen}$^\text{c}$\\
  $^\text{a}${\it\small Institute of Mathematics, Hangzhou Dianzi University, Hangzhou 310018, PR~China}\\
  $^\text{b}${\it\small Department of Mathematics, Anhui University of
Technology, Ma Anshan 243032, PR~China}\\
  $^\text{c}${\it\small Department of Applied Mathematics, Zhejiang University of
Technology, Hangzhou, Zhejiang 310023, PR~China}}

\date{}
\maketitle

\mbox{}\hrule\mbox\\
\begin{abstract}

We derive the spectral density of the equiprobable mixture of two
random density matrices of a two-level quantum system. We also work
out the spectral density of mixture under the so-called quantum
addition rule. We use the spectral densities to calculate the
average entropy of mixtures of random density matrices, and  show
that the average entropy of the arithmetic-mean-state of $n$ qubit
density matrices randomly chosen from the Hilbert-Schmidt ensemble
is never decreasing with the number $n$. We also get the exact value
of the average squared fidelity. Some conjectures and open problems
related to von Neumann entropy are also
proposed.

\end{abstract}
\maketitle \mbox{}\hrule\mbox\\
\newpage
\tableofcontents

\newpage

\section{Introduction}

In the early 1950s, physicists had reached the limits of
deterministic analytical techniques for studying the energy spectra
of heavy atoms undergoing slow nuclear reactions. It is well-known
that a random matrix with appropriate symmetries might serve as a
suitable model for the Hamiltonian of the quantum mechanical system
that describes the reaction \cite{Wigner1958}. The eigenvalues of
this random matrix model the possible energy levels of the system
\cite{Mehta2004}. In quantum statistical mechanics, the canonical
states of the system under consideration are the reduced density
matrices of the uniform states on a subspace of system and
environment. Moreover, such reduced density matrices can be realized
by Wishart matrix ensemble \cite{Zyczkowski2011}. Thus
investigations by using random matrix theoretical techniques can
lead to deeper insightful perspectives on some problems in Quantum
Information Theory
\cite{Hayden2006,Hastings2009,Aubrun2012,Aubrun2014,Oszmaniec2016,Puchala2016,Nechita2016}.
In fact, most works using RMT as a tool to study quantum information
theory are concentrated on the limiting density and their
asymptotics. In stark contrast, researchers obtained an exact
probability distribution of eigenvalues of a multipartite random
quantum state via deep mathematical tools such as symplectic
geometric method albeit the used definition of Duistermaat-Heckman
measure is very abstract and difficult
\cite{Christandl2014,Walter2014}. Besides, the authors conducted
exact and asymptotic spectral analysis of the difference between two
random mixed quantum states \cite{Mejia2017}. Non-asymptotic results
about average quantum coherence for a random quantum state
\cite{Uttam2016,Lin2017a,Lin2017b} and its typicality were obtained
recently. Motivated by the connection of the works
\cite{Christandl2014,Walter2014} and Horn's problem
\cite{Zhang2017}, we focus the spectral analysis of mixture of
several random states in a two-level system. Although the spectral
analysis of superposition of random pure states were performed
recently \cite{ZPNC2011,CFF2013}, the topic about the spectral
densities for mixtures of random density matrices from two quantum
state ensembles is rarely discussed previously.

Along this line, we will make an attempt toward exact spectral
analysis of two kinds of mixtures of two random density matrices for
qubits: a) equiprobable mixture of two random density matrices,
based on the results obtained in Ref. \cite{Zhang2017}, and b)
mixture of two random density matrices under the quantum addition
rule (see Definition~\ref{def:additionrule}, \cite{Audenaert2016}).
To the best of our knowledge, such kind of spectral analysis for
mixture of random states is rarely conducted, in particular the
spectral density under the quantum addition rule. The aim of this
work is to analyze properties of a generic quantum state on
two-dimensional Hilbert space. For two random states chosen from two
unitary orbits, each distributed according to Haar measure over
$\rS\rU(2)$, we derive the spectral density of the equiprobable
mixture of both random density matrices for qubits, and the spectral
density of mixture of both random density matrices under the quantum
addition rule. When they are distributed according to the
Hilbert-Schmidt measure in the set $\density{\complex^2}$, i.e., the
set of all $2\times 2$ density matrices, of quantum states of
dimension two, we can calculate the average entropy of ensemble
generated by two kinds of mixtures. We also study entropy inequality
under the quantum addition rule.

The paper is organized as follows: In
Section~\ref{sect:preliminaries}, we recall some useful facts about
a qubit. Then we present our main results with their proofs in
Section~\ref{sect:mainresults}. Specifically, we obtain the spectral
densities of two kinds of mixtures of two qubit density matrices:
(a) the equiprobable mixture and (b) the mixture under the quantum
addition rule. By using the relationship between an eigenvalue of a
qubit density matrix and the length of its Bloch vector
representation, we get compact forms (Theorem~\ref{th:rr1r2} and
Theorem~\ref{th:qr-density}) of corresponding spectral densities. We
also investigate a quantum Jensen-Shannon divergence-like quantity,
based on the mixture of two random density matrices under the
quantum addition rule. It provide a universal lower bound for the
quantum Jensen-Shannon divergence. However our numerical experiments
show that such lower bound cannot define a true metric. Next, in
Section~\ref{sect:app}, we use the obtained results in the last
section to calculate the average entropies of mixtures of two random
density matrices in a two-level quantum system. We show that the
average entropy of the arithmetic-mean-state of $n$ qubit density
matrices being randomly chosen from the Hilbert-Schmidt ensemble is
never decreasing with $n$. As further illumination of our results,
we make an attempt to explain why 'mixing reduces coherence'. We
also work out the exact value of the average squared fidelity,
studied intensively by K. \.{Z}yczkowski. Finally, we conclude this
paper with some remarks and open problems.

\section{Preliminaries}\label{sect:preliminaries}

To begin with, we recall some facts about a qubit. Any qubit density
matrix can be represented as
\begin{eqnarray}\label{eq:Bloch-vector}
\rho(\boldsymbol{r}) = \frac12(\mathbb{I}_2+\boldsymbol{r}\cdot
\boldsymbol{\sigma}),
\end{eqnarray}
where $\boldsymbol{r}=(r_x,r_y,r_z)\in\real^3$ is the Bloch vector
with $r:=\abs{\boldsymbol{r}}\leqslant 1$, and
$\boldsymbol{\sigma}=(\sigma_x,\sigma_y,\sigma_z)$. Here
\begin{eqnarray*}
\sigma_x=\Pa{\begin{array}{cc}
               0 & 1 \\
               1 & 0
             \end{array}
}, \sigma_y=\Pa{\begin{array}{cc}
                  0 & -\mathrm{i} \\
                  \mathrm{i} & 0
                \end{array}
}, \sigma_z=\Pa{\begin{array}{cc}
                  1 & 0 \\
                  0 & -1
                \end{array}
}
\end{eqnarray*}
are three Pauli matrices. The eigenvalues of a qubit density matrix
are given by: $\lambda_\pm(\rho)=\frac12\Pa{1\pm r}$, where
$r\in[0,1]$. This leads to the von Neumann entropy of the qubit
$\rho(\boldsymbol{r})$ of Bloch vector $\boldsymbol{r}$:
\begin{eqnarray}
\rS(\rho(\boldsymbol{r}))= \rH_2\Pa{\frac{1-r}2}:=\Phi(r),
\end{eqnarray}
where $\rH_2\Pa{p}= -p\log_2p -(1-p)\log_2(1-p)$ for $p\in[0,1]$ is
the binary entropy function.

Note that the maximal eigenvalue $\lambda_+(\rho)$ for a random
qubit density matrix, induced from taking partial trace over a
Haar-distributed pure two-qubit state, is subject to the following
distribution \cite{Christandl2014}:
\begin{eqnarray}
\dif P(x) = 24\Pa{x-\frac12}^2\dif x,
\end{eqnarray}
where $x:=\lambda_+(\rho)\in[1/2,1]$. Since
$\lambda_+(\rho)=\frac12(1+r)$, it follows that the probability
density of the length $r$ of Bloch vector of a random qubit
$\rho(\boldsymbol{r})$ is summarized into the following proposition.
\begin{prop}\label{prop:original}
The probability density for the length $r$ of the Bloch vector
$\boldsymbol{r}$ in the Bloch representation \eqref{eq:Bloch-vector}
of a random qubit $\rho$ by partial-tracing other qubit system over
a Haar-distributed pure two-qubit state, is given by
\begin{eqnarray}
p_r(r)=3r^2,\quad r\in[0,1].
\end{eqnarray}
Furthermore, the probability distribution of Bloch vector
$\boldsymbol{r}$ is given by the formula: $p(\boldsymbol{r})[\dif
\boldsymbol{r}]=3r^2\dif r\times
\frac1{4\pi}\delta(1-\abs{\boldsymbol{u}})[\dif \boldsymbol{u}]$,
where $\delta$ is the Dirac delta function and $[\dif
\boldsymbol{u}]$ is the Lebesgue volume element in $\real^3$.
\end{prop}

\section{Main results}\label{sect:mainresults}

\subsection{The spectral density of equiprobable mixture of two qubit states}

For $w\in[0,1]$, we have the probabilistic mixture of two density
matrices in a two-level system:
$\rho_w(\boldsymbol{r})=w\rho(\boldsymbol{r}_1)+(1-w)\rho(\boldsymbol{r}_2)$.
In particular, for $w=1/2$, we have the equiprobable mixture, that
is,
$\rho(\boldsymbol{r})=\frac{\rho(\boldsymbol{r}_1)+\rho(\boldsymbol{r}_2)}2$,
hence $\boldsymbol{r}=\frac{\boldsymbol{r}_1+\boldsymbol{r}_2}2$. As
in \cite{Zhang2017}, denote $\mu,\nu\in[0,1/2]$ the minimal
eigenvalues of two qubit states $\rho(\boldsymbol{r}_1)$ and
$\rho(\boldsymbol{r}_2)$, respectively. Then we have
$\mu=\frac{1-r_1}2$ and $\nu=\frac{1-r_2}2$ for $r_1,r_2\in[0,1]$.
Denote $\cO_\mu:=\Set{U\diag(1-\mu,\mu)U^\dagger: U\in\rS\rU(2)}$.
We consider the equiprobable mixture of two random density matrices
$\rho(\boldsymbol{r}_1)\in\cO_\mu$ and $\rho(\boldsymbol{r}_2)\in
\cO_\nu$. In \cite{Zhang2017}, we have already derived the
analytical formula for the spectral density of such equiprobable
mixture. This result can be summarized into the following
proposition.
\begin{prop}[\cite{Zhang2017}]\label{prop:lmn}
The probability density function of an eigenvalue $\lambda$ of the
equiprobable mixture of two random density matrices, chosen
uniformly from respective unitary orbits $\cO_\mu$ and $\cO_\nu$
with $\mu,\nu$ are fixed in $(0,1/2)$, is given by
\begin{eqnarray}\label{eq:lambdamunu}
p(\lambda|\mu,\nu) =
\frac{\abs{\lambda-\frac12}}{\Pa{\frac12-\mu}\Pa{\frac12-\nu}},
\end{eqnarray}
where $\lambda \in[T_0,T_1]\cup[1-T_1,1-T_0]$. Here
$T_0:=\frac{\mu+\nu}2$ and $T_1:=\frac{1-\abs{\mu-\nu}}2$.
\end{prop}
Note that $\lambda \in[T_0,T_1]\cup[1-T_1,1-T_0]$ indicates that the
domain of an eigenvalue of the equiprobable mixture:
$\frac12\Pa{U\diag(1-\mu,\mu)U^\dagger+V\diag(1-\nu,\nu)V^\dagger}$,
where $U,V\in\rS\rU(2)$.

Given two random density matrices $\rho(\boldsymbol{r}_1)\in
\cO_\mu$ and $\rho(\boldsymbol{r}_2)\in\cO_\nu$. We also see that
the eigenvalues of the mixture $\rho(\boldsymbol{r})$ are given by
$\lambda=\frac{1\pm r}2$ . The sign $\pm$ depends on the
relationship between $\lambda$ and $1/2$. Indeed,
$\lambda=\frac{1+r}2$ if $\lambda\geqslant 1/2$;
$\lambda=\frac{1-r}2$ if $\lambda\leqslant 1/2$. By using the triple
$(r_1,r_2,r)$ instead of $(\mu,\nu,\lambda)$, we have the following
result.
\begin{thrm}\label{th:rr1r2}
The conditional probability density function of the length $r$ of
the Bloch vector $\boldsymbol{r}$ of the equiprobable mixture:
$\rho(\boldsymbol{r})=\frac{\rho(\boldsymbol{r}_1)+\rho(\boldsymbol{r}_2)}2$,
where $r_1,r_2\in(0,1)$ are fixed, is given by
\begin{eqnarray}
p(r|r_1,r_2) = \frac{2r}{r_1r_2},
\end{eqnarray}
where $r\in[r_-,r_+]$ with $r_-:=\frac{\abs{r_1-r_2}}2$ and
$r_+:=\frac{r_1+r_2}2$.
\end{thrm}
Denote by $\theta$ the angle between Bloch vectors
$\boldsymbol{r}_1$ and $\boldsymbol{r}_2$ from two random density
matrices $\rho(\boldsymbol{r}_1)\in \cO_\mu$ and
$\rho(\boldsymbol{r}_2)\in\cO_\nu$, respectively. Apparently
$\theta\in[0,\pi]$. Since
$\rho(\boldsymbol{r})=\frac{\rho(\boldsymbol{r}_1)+\rho(\boldsymbol{r}_2)}2$,
i.e., $\boldsymbol{r}=\frac{\boldsymbol{r}_1+\boldsymbol{r}_2}2$, it
follows that $r=\frac12\sqrt{r^2_1+r^2_2+2r_1r_2\cos\theta}$, where
$\theta\in[0,\pi]$.

Clearly the rhs is the invertible function of the argument $\theta$
defined over $[0,\pi]$ when $r_1$ and $r_2$ are fixed. In view of
this, we see that the angle between two random Bloch vectors has the
following probability density:
\begin{eqnarray}\label{eq:theta-density}
f(\theta)=\frac12\sin\theta,\quad \theta\in[0,\pi].
\end{eqnarray}

\subsection{The quantum addition rule for two qubit states}

Shannon's Entropy Power Inequality mainly deals with the concavity
of an entropy function of a continuous random variable under the
scaled addition rule: $f\Pa{\sqrt{w}X+\sqrt{1-w}Y}\geqslant
wf(X)+(1-w)f(Y)$, where $w\in[0,1]$ and $X,Y$ are continuous random
variables and the function $f$ is either the differential entropy or
the entropy power \cite{Konig2014}. Some generalizations in the
quantum regime along this line are obtained recently. For instance,
quantum analogues of these inequalities for continuous-variable
quantum systems are obtained, where $X$ and $Y$ are replaced by
bosonic fields and the addition rule is the action of a beam
splitter with transmissivity $w$ on those fields
\cite{Konig2014,Palma2014}. Similarly, Audenaert \emph{et al}
establish a class of entropy power inequality analogs for qudits.
The addition rule used in these inequalities is given by the
so-called \emph{partial swap channel} \cite{Audenaert2016}. Let us
recall some notions we will use in this paper.

Let $\set{\ket{j}:j=1,\ldots,d}$ be the standard basis of
$\complex^d$. Then $\set{\ket{ij}:i,j=1,\ldots,d}$ is an orthonormal
basis of $\complex^d\ot\complex^d$. Denote by $\density{\complex^d}$
the set of all density matrices on $\complex^d$. The swap operator
$S\in\rU(\complex^d\ot\complex^d)$, the unitary group on
$\complex^d\ot\complex^d$, is defined through its action on the
basis vectors $\ket{ij}$ as follows: $S\ket{ij}=\ket{ji}$ for all
$i,j=1,\ldots,d$. Explicitly, the swap operator can be rewritten as
$S=\sum^{d}_{i,j=1}\out{ij}{ji}$. From the definition of the swap
operator, we see that $S$ is self-adjoint and unitary. Audenaert
\emph{et al} defined a qudit partial swap operator as a unitary
interpolation between the identity and the swap operator in
\cite{Audenaert2016}.
\begin{definition}[Partial swap operator]
For $t\in[0,1]$, the partial swap operator
$U_t\in\rU(\complex^d\ot\complex^d)$ is the unitary operator $U_t :=
\sqrt{t}\mathbb{I}_d\ot\mathbb{I}_d + \sqrt{1-t}\mathrm{i}S$.
\end{definition}
It is easily seen that the matrix representation of the partial swap
operator for the two-level system is
\begin{eqnarray*}
U_t &=&\Br{\begin{array}{cccc}
            \sqrt{t}+\mathrm{i}\sqrt{1-t} & 0 & 0 & 0 \\
            0 & \sqrt{t} & \mathrm{i}\sqrt{1-t} & 0 \\
            0 & \mathrm{i}\sqrt{1-t} & \sqrt{t} & 0 \\
            0 & 0 & 0 & \sqrt{t}+\mathrm{i}\sqrt{1-t}
          \end{array}
}.
\end{eqnarray*}
When $t=1/2$, we have
\begin{eqnarray*}
U_{1/2} &=& \frac1{\sqrt{2}}\Br{\begin{array}{cccc}
            1+\mathrm{i} & 0 & 0 & 0 \\
            0 & 1 & \mathrm{i} & 0 \\
            0 & \mathrm{i} & 1 & 0 \\
            0 & 0 & 0 & 1+\mathrm{i}
          \end{array}
},
\end{eqnarray*}
Consider a family of CPTP maps
$\cE_t:\density{\complex^d\ot\complex^d}\to\density{\complex^d}$
parameterized by $t\in[0,1]$. It is defined in terms of the partial
swap operator $U_t$ in the above definition. For any
$\rho_{12}\in\density{\complex^d\ot\complex^d}$, let
$\cE_t(\rho_{12}) = \Ptr{2}{U_t\rho_{12}U^\dagger_t}$. Denote
$\widehat \cE_t(\rho_{12}) = \Ptr{1}{U_t\rho_{12}U^\dagger_t}$. We
are particularly interested in the case where the input state
$\rho_{12}$ is a product state, i.e. $\rho_{12}=\rho_1\ot\rho_2$ for
$\rho_j\in\density{\complex^d}, j=1,2$. Apparently,
$S(\rho_1\ot\rho_2)S=\rho_2\ot\rho_1$. Now
\begin{eqnarray}
U_t(\rho_1\ot\rho_2)U^\dagger_t = t\rho_1\ot\rho_2 +
(1-t)\rho_2\ot\rho_1+\mathrm{i}\sqrt{t(1-t)}[S,\rho_1\ot\rho_2].
\end{eqnarray}
From this, we see easily that
\begin{eqnarray}
\cE_t(\rho_1\ot\rho_2) &=& t\rho_1+(1-t)\rho_2
-\mathrm{i}\sqrt{t(1-t)}[\rho_1,\rho_2],\\
\widehat\cE_t(\rho_1\ot\rho_2) &=& t\rho_2+(1-t)\rho_1
+\mathrm{i}\sqrt{t(1-t)}[\rho_1,\rho_2].
\end{eqnarray}
In particular, for $t=\frac12$, we get
\begin{eqnarray}
\cE_{1/2}(\rho_1\ot\rho_2) &=& \frac12\Pa{\rho_1+\rho_2
-\mathrm{i}[\rho_1,\rho_2]},\\
\widehat\cE_{1/2}(\rho_1\ot\rho_2) &=& \frac12\Pa{\rho_1+\rho_2
+\mathrm{i}[\rho_1,\rho_2]}.
\end{eqnarray}

\begin{definition}[Quantum addition rule]\label{def:additionrule}
For any $t\in[0,1]$ and any $\rho_1,\rho_2\in\density{\complex^d}$,
we define $\rho_1\boxplus_t \rho_2 := \cE_t(\rho_1\ot\rho_2)$. It is
trivial that $\rho_2\boxplus_t\rho_1=\widehat
\cE_t(\rho_1\ot\rho_2)$.
\end{definition}
Denote $g_d(t):=\rS(\rho_1\boxplus_t \rho_2)+\rS(\rho_2\boxplus_t
\rho_1) - \rS(\rho_1)-\rS(\rho_2)$, where $\rS(\rho)$ is the von
Neumann entropy of $\rho$. This can be viewed as the mutual
information between two $d$-level subsystems after performing
$\cE_t$ when their composite system lives in the product form. In
other words, such quantity $g_d(t)$ stands for the correlative power
of the partial swap channel $\cE_t$, we conjecture
$\max_{t\in[0,1]}g_d(t)=g_d(1/2)$ for any $d\geqslant 2$, that is,
the correlative power of the partial swap channel achieves its
maximum at $t=1/2$. We next give a positive answer to this
conjecture in the qubit case. The proof for the qudit case is
expected.
\begin{prop}\label{prop:CP-PSC}
For any $t\in[0,1]$ and any two qubit density matrices
$\rho_1,\rho_2\in\density{\complex^2}$, we have
\begin{eqnarray}
\max_{t\in[0,1]}g_2(t)=g_2(1/2).
\end{eqnarray}
That is,
\begin{eqnarray}
\max_{t\in[0,1]} \Pa{\rS(\rho_1\boxplus_t
\rho_2)+\rS(\rho_2\boxplus_t \rho_1)}=2\rS(\rho_1\boxplus_{1/2}
\rho_2).
\end{eqnarray}
\end{prop}

\begin{proof}
Denote $s(t)=\rS(\rho_1\boxplus_t \rho_2)+\rS(\rho_2\boxplus_t
\rho_1)$. Thus $s(t)=\Phi(r_{12}(t))+\Phi(r_{21}(t))$, where
$$
\Phi(x) = -\frac{1+x}2\log_2\frac{1+x}2 -
\frac{1-x}2\log_2\frac{1-x}2,
$$
and $r_{12}(t)=\abs{\boldsymbol{r}(\rho_1\boxplus_t \rho_2)},
r_{21}(t) = \abs{\boldsymbol{r}(\rho_2\boxplus_t \rho_1)}$. Then
$$
s'(t) = \frac12 r'_{12}(t)\log_2\frac{1-r_{12}(t)}{1+r_{12}(t)}
+\frac12 r'_{21}(t)\log_2\frac{1- r_{21}(t)}{1+ r_{21}(t)}
$$
and
\begin{eqnarray*}
s''(t) &=& \frac12 r''_{12}(t)\log_2\frac{1-r_{12}(t)}{1+r_{12}(t)}
+\frac12 r''_{21}(t)\log_2\frac{1-r_{21}(t)}{1+r_{21}(t)}\\
&&-\frac1{\ln2}\Pa{\frac{[r'_{12}(t)]^2}{1-r_{12}(t)^2} +
\frac{[r'_{21}(t)]^2}{1-r_{21}(t)^2}}.
\end{eqnarray*}
Note that $0\leqslant r_{12}(t)\leqslant tr_1+(1-t)r_2\leqslant 1$
and $0\leqslant r_{21}(t)\leqslant (1-t)r_1+tr_2\leqslant 1$ because
$r_1,r_2\in[0,1]$. Thus
$$
-\frac1{\ln2}\Pa{\frac{[r'_{12}(t)]^2}{1-r_{12}(t)^2} +
\frac{[r'_{21}(t)]^2}{1-r_{21}(t)^2}}<0.
$$
Denote $\alpha:=2r_1r_2\cos\theta+r^2_1r^2_2\sin^2\theta$ and
$\varphi(t)=(r^2_1+r^2_2-\alpha)t^2+\alpha t$. Hence
$$
r_{12}(t)=\sqrt{\varphi(t)-2r^2_2t+r^2_2},\quad
r_{21}(t)=\sqrt{\varphi(t)-2r^2_1t+r^2_1}.
$$
This implies that $r'_{12}(t) =
\frac{\varphi'(t)-2r^2_2}{2r_{12}(t)}$. Based on this, we obtain
\begin{eqnarray*}
r''_{12}(t) = \frac{2\varphi''(t)r^2_{12}(t) -
\Pa{\varphi'(t)-2r^2_2}^2}{4r^3_{12}(t)}.
\end{eqnarray*}
By using $\varphi'(t),\varphi''(t)$, and $r^2_{12}(t)$, we have
\begin{eqnarray*}
2\varphi''(t)r^2_{12}(t) - \Pa{\varphi'(t)-2r^2_2}^2
&=& 4r^2_2(r^2_1+r^2_2-\alpha) - (\alpha-2r^2_2)^2\\
&=&r^2_1r^2_2\Pa{4-(2\cos\theta+r_1r_2\sin^2\theta)^2}\geqslant0,
\end{eqnarray*}
implying $r''_{12}(t)> 0$. Similarly, we see that $r''_{21}(t)
> 0$. In summary, we get that $s''(t) <0$.
That is, $s(t)$ is the strict concave function over $[0,1]$. It is
easily seen that $r_{12}(1/2)=r_{21}(1/2)=\sqrt{\varphi(1/2)}$. Now
$r'_{12}(1/2)=\frac{\varphi'(1/2)-2r^2_2}{2r_{12}(1/2)}$ and
$r'_{21}(1/2)=\frac{\varphi'(1/2)-2r^2_1}{2r_{12}(1/2)}$. Apparently
$\varphi'(1/2)=r^2_1+r^2_2$ by the expression of $\varphi$.
Substituting $r_{12}(1/2),r'_{12}(1/2),\varphi'(1/2)$ into the
expression of $s'(1/2)$ gives rise to the result: $s'(1/2)=0$.
Therefore the maximum of $s(t)$ on $[0,1]$ is taken at $1/2$, i.e.,
$\max_{t\in[0,1]}s(t)=s(1/2)$. This is equivalent to the desired
conclusion $\max_{t\in[0,1]}g_2(t)=g_2(1/2)$. We are done.
\end{proof}

The result in Proposition~\ref{prop:CP-PSC} can be viewed as another
proof of the following inequality:
$$
\rS(\rho_1\boxplus_{1/2}\rho_2)\geqslant
\frac12\rS(\rho_1)+\frac12\rS(\rho_2).
$$
We also have the following interesting result:
\begin{prop}\label{prop:EI}
In a two-level system, it holds that
\begin{eqnarray}
w\rS(\rho_1)+(1-w)\rS(\rho_2)&\leqslant&\rS(\rho_1\boxplus_w\rho_2)\label{eq:a1}\\
&\leqslant& \rS\Pa{w\rho_1+(1-w)\rho_2}\label{eq:a2}
\end{eqnarray}
for any weight $w\in[0,1]$. In particular, for $w=1/2$, we have
\begin{eqnarray}\label{eq:1/2}
\frac{\rS(\rho_1)+\rS(\rho_2)}2\leqslant\rS(\rho_1\boxplus_{1/2}\rho_2)
\leqslant \rS\Pa{\frac{\rho_1+\rho_2}2}.
\end{eqnarray}
\end{prop}

\begin{proof}
The inequality in \eqref{eq:a1} is obtained in \cite{Audenaert2016}.
In order to prove the second one in \eqref{eq:a2}, we use Bloch
representation of a density matrix for a qubit state, for
$w\in[0,1]$,
\begin{eqnarray*}
\rho(\boldsymbol{r}(w))&:=&w\rho(\boldsymbol{r}_1)+(1-w)\rho(\boldsymbol{r}_2),\\
\rho(\hat{\boldsymbol{r}}(w))&:=&\rho(\boldsymbol{r}_1)\boxplus_w\rho(\boldsymbol{r}_2).
\end{eqnarray*}
We see that
\begin{eqnarray*}
\boldsymbol{r}(w) &=& w\boldsymbol{r}_1+(1-w)\boldsymbol{r}_2,\\
\hat{\boldsymbol{r}}(w) &=& w\boldsymbol{r}_1+(1-w)\boldsymbol{r}_2
+\sqrt{w(1-w)}\boldsymbol{r}_1\times\boldsymbol{r}_2.
\end{eqnarray*}
Let $\theta\in[0,\pi]$ is the angle between vectors
$\boldsymbol{r}_1$ and $\boldsymbol{r}_2$. Thus
\begin{eqnarray*}
r(w)^2 &=& w^2r^2_1+(1-w)^2r^2_2+2w(1-w)r_1r_2\cos\theta, \\
\hat r(w)^2 &=& w^2r^2_1+(1-w)^2r^2_2+2w(1-w)r_1r_2\cos\theta
+w(1-w)r^2_1r^2_2\sin^2\theta.
\end{eqnarray*}
Since
\begin{eqnarray*}
\Phi(r(w))&=&\rS\Pa{w\rho(\boldsymbol{r}_1)+(1-w)\rho(\boldsymbol{r}_2)},\\
\Phi(\hat r(w))&=&
\rS(\rho(\boldsymbol{r}_1)\boxplus_w\rho(\boldsymbol{r}_2)).
\end{eqnarray*}
Clearly $1\geqslant\hat r(w)\geqslant r(w)\geqslant0$. Note that
$\Phi(x)$ is decreasing over $[0,1]$ since
$$
\Phi'(x)=\frac12\Br{\log_2(1-x)-\log_2(1+x)}\leqslant0,
$$
it follows that $\Phi(\hat r(w)) \leqslant \Phi(r(w))$ for any
weight $w\in[0,1]$. We get the desired second inequality.
Consequently, for $w=1/2$, thus we get \eqref{eq:1/2}.
\end{proof}

We believe that Proposition~\ref{prop:EI} can be generalized to a
qudit system. In fact, \eqref{eq:a1} holds for a qudit system. We
\emph{conjecture} that \eqref{eq:a2} holds also for a qudit system,
i.e.,
\begin{eqnarray}
\rS(\rho_1\boxplus_w\rho_2)\leqslant
\rS\Pa{w\rho_1+(1-w)\rho_2},\quad \forall w\in[0,1].
\end{eqnarray}

\subsection{A lower bound for the quantum Jensen-Shannon divergence}

Recently, Majtey \emph{et al} \cite{Majtey2005} introduced a quantum
analog of Jensen-Shannon divergence, called \emph{quantum
Jensen-Shannon divergence (QJSD)}, as a measure of
distinguishability between mixed quantum states. Since QJSD shares
most of the physically relevant properties with the relative
entropy, it is considered to be the ``good" quantum
distinguishability measure.

Recall that the quantum Jensen-Shannon divergence (QJSD) is defined
by
\begin{eqnarray}
J(\rho_1,\rho_2):= \rS\Pa{\frac{\rho_1 +\rho_2}2}-
\frac{\rS(\rho_1)+\rS(\rho_2)}2.
\end{eqnarray}
Lamberti \emph{et al} \cite{Lamberti2008} discussed the metric
character of QJSD. They proposed a \emph{conjecture} related to the
quantum Jensen-Shannon divergence: the following distance, based on
QJSD, is the true metric
\begin{eqnarray}
\bD_{J}(\rho_1,\rho_2) =\sqrt{J(\rho_1,\rho_2)}.
\end{eqnarray}
Note that this conjecture is proven to be true for qubit systems and
pure qudit systems \cite{Briet2009}. Numerical evidence supports it
for mixed qudit systems.

By employing the quantum addition rule, we can provide a lower bound
for the QJSD. Denote
\begin{eqnarray}
\hat{J}(\rho_1,\rho_2):= \rS(\rho_1\boxplus_{1/2} \rho_2)-
\frac{\rS(\rho_1)+\rS(\rho_2)}2.
\end{eqnarray}
With the above notation, we define a new distance:
\begin{eqnarray}
\bD_{\hat{J}}(\rho_1,\rho_2) =\sqrt{\hat{J}(\rho_1,\rho_2)}.
\end{eqnarray}
Note that $\rS\Pa{\frac{\rho_1
+\rho_2}2}\geqslant\rS(\rho_1\boxplus_{1/2} \rho_2)$, we see that
$$
J(\rho_1,\rho_2)\geqslant\hat J(\rho_1,\rho_2)
$$
and thus
$$
\bD_{J}(\rho_1,\rho_2)\geqslant \bD_{\hat J}(\rho_1,\rho_2).
$$
Similarly, we can study the lower bound based on the quantum
addition rule for two-level systems in
Definition~\ref{def:additionrule}. We expect the quantity $\bD_{\hat
J}$ to be the true metric. But in fact it is not, as suggested in
the following figures. Clearly,
\begin{eqnarray}
\bD_{\hat J}(\rho(\boldsymbol{r}_1),\rho(\boldsymbol{r}_2)) =
\sqrt{\Phi(\hat r_{12}) - \frac{\Phi(r_1)+\Phi(r_2)}2},
\end{eqnarray}
where $\hat
r_{12}=\frac12\sqrt{r^2_1+r^2_2+2r_1r_2\cos\theta+r^2_1r^2_2\sin^2\theta}$
for $r_1,r_2\in(0,1)$ and $\theta\in[0,\pi]$.

Denote
\begin{eqnarray}
\boldsymbol{\Delta}(\boldsymbol{r}_1,\boldsymbol{r}_2,\boldsymbol{r}_3)
:=\bD_{\hat
J}(\rho(\boldsymbol{r}_1),\rho(\boldsymbol{r}_2))+\bD_{\hat
J}(\rho(\boldsymbol{r}_1),\rho(\boldsymbol{r}_3)) -\bD_{\hat
J}(\rho(\boldsymbol{r}_2),\rho(\boldsymbol{r}_3)).
\end{eqnarray}
Our numerical experiments, as demonstrated in the figures, i.e.,
Fig.~\ref{fig:Fig1} and Fig.~\ref{fig:Fig2}, show that $\bD_{\hat
J}(\rho(\boldsymbol{r}_1),\rho(\boldsymbol{r}_2))$ does not satisfy
the triangle inequality by randomly generating the thousands of
qubits. That is, there exist some random samples such that
$\boldsymbol{\Delta}<0$.

In FIG.~\ref{fig:Fig1}, there are some states which do not satisfy
the triangle inequality, for example
\begin{eqnarray*}
\boldsymbol{r}_1&=&[0.594637,-0.562167,
-0.402354]^\t,\\
\boldsymbol{r}_2&=&[0.246183,-0.755573,
0.593725]^\t,\\
\boldsymbol{r}_3&=&[0.190508,-0.0792096, -0.855743]^\t,
\end{eqnarray*}
we have $\boldsymbol{\Delta}=-0.0820814<0$ for the above triple.
This means that $\bD_{\hat J}$ is not the true metric over
$\density{\complex^2}$.

\begin{figure}[ht]
\centering {\begin{minipage}[b]{0.8\linewidth}
\includegraphics[width=1\textwidth]{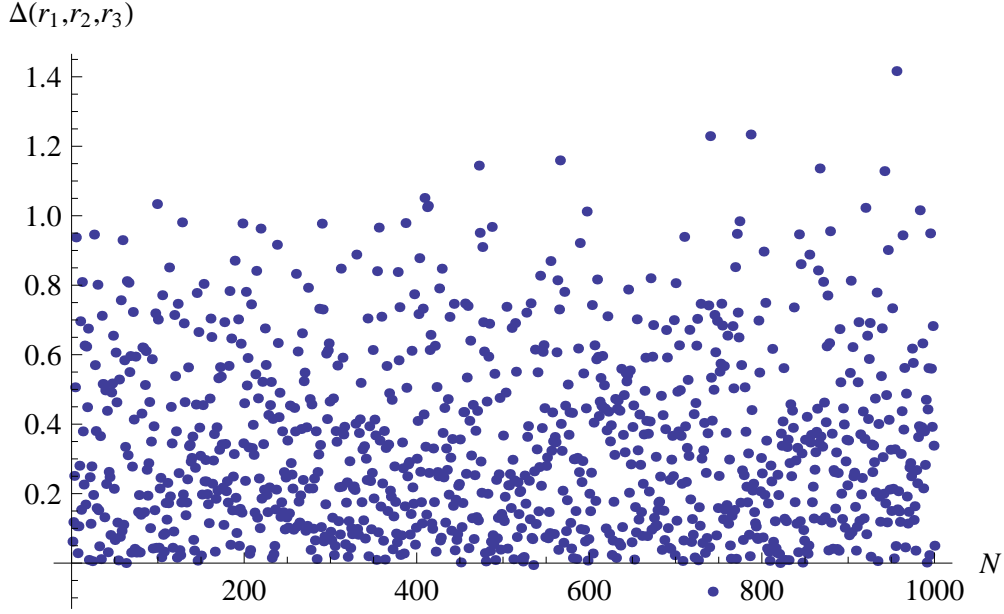}
\end{minipage}}
\caption{(Color online) Random testing of the violation of the
triangular inequality for $\bD_{\hat J}$ by using \emph{mixed} qubit
states. The horizontal axis $N$ represents the numbers of
$\boldsymbol{\Delta}(\boldsymbol{r}_1,\boldsymbol{r}_2,\boldsymbol{r}_3)$
which are calculated here.
$\rho(\boldsymbol{r}_1),\rho(\boldsymbol{r}_2)$ and
$\rho(\boldsymbol{r}_3)$ are generally not pure states. The points
under the horizontal axis indicate the violation of the triangular
inequality by triples of three qubit states.} \label{fig:Fig1}
\end{figure}

\begin{figure}[ht]
\centering {\begin{minipage}[b]{0.8\linewidth}
\includegraphics[width=1\textwidth]{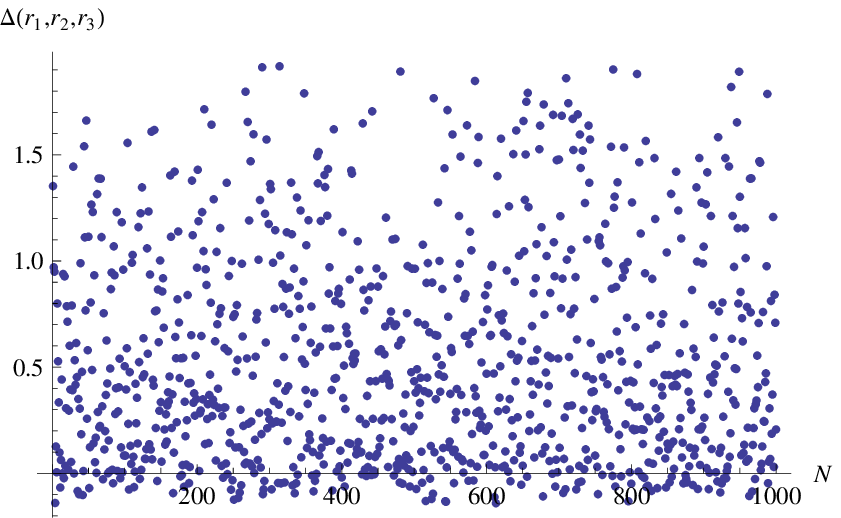}
\end{minipage}}
\caption{(Color online) Random testing of the violation of the
triangular inequality for $\bD_{\hat J}$ by using \emph{pure} qubit
states. The horizontal axis $N$ represents the numbers of
$\boldsymbol{\Delta}(\boldsymbol{r}_1,\boldsymbol{r}_2,\boldsymbol{r}_3)$
which are calculated here. All
$\rho(\boldsymbol{r}_1),\rho(\boldsymbol{r}_2)$ and
$\rho(\boldsymbol{r}_3)$ are \emph{pure} states). The points under
the horizontal axis indicate the violation of the triangular
inequality by triples of three qubit \emph{pure} states.}
\label{fig:Fig2}
\end{figure}

In FIG.~\ref{fig:Fig2}, as suggested in this figure, the inequality
$\boldsymbol{\Delta}\geqslant0$, i.e., the triangle inequality, is
violated by a lot of random samples representing pure qubit states.
Again, $\bD_{\hat J}$ is not the true metric over the set of all
pure qubit states.

\begin{figure}[ht]
\centering {\begin{minipage}[b]{0.8\linewidth}
\includegraphics[width=1\textwidth]{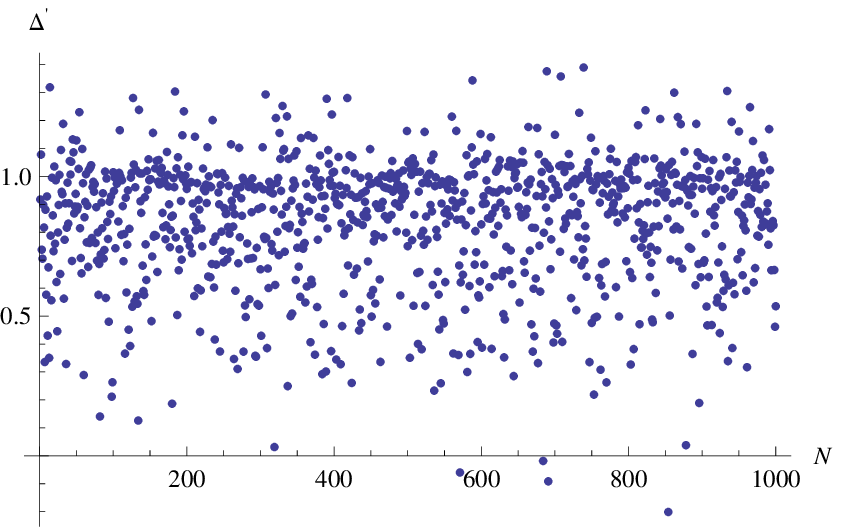}
\end{minipage}}
\caption{(Color online) Random testing of the violation of the
triangular inequality for $\bD^2_{\hat J}$ by using \emph{mixed}
qubit states. The horizontal axis $N$ represents the numbers of
$\boldsymbol{\Delta}'$ which are calculated here.
$\rho_1=\rho(\boldsymbol{r}_1),\rho_2=\rho(\boldsymbol{r}_2)$ and
$\rho=\rho(\boldsymbol{r}_3)$ are generally not pure states. The
points under the horizontal axis indicate the violation of the
triangular inequality by triples of three qubit states.}
\label{fig:Fig3}
\end{figure}

\begin{figure}[ht]
\centering {\begin{minipage}[b]{0.8\linewidth}
\includegraphics[width=1\textwidth]{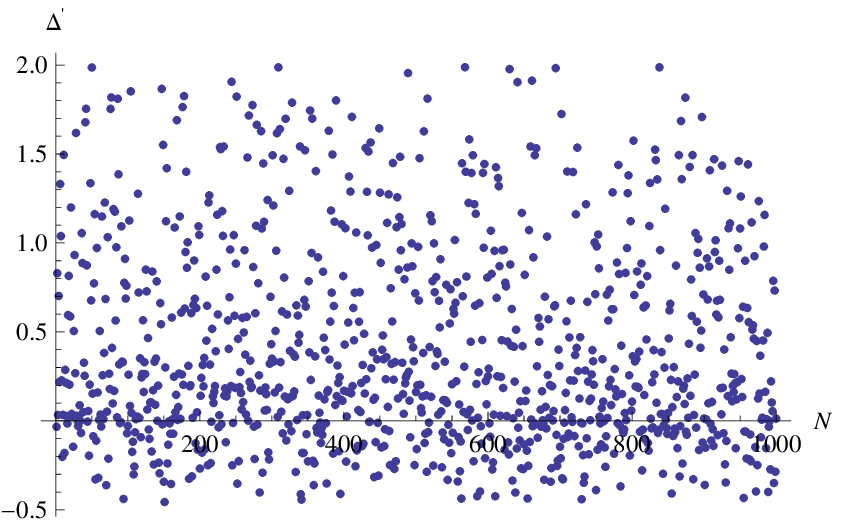}
\end{minipage}}
\caption{(Color online) Random testing of the violation of the
triangular inequality for $\bD^2_{\hat J}$ by using \emph{pure}
qubit states. The horizontal axis $N$ represents the numbers of
$\boldsymbol{\Delta}'$ which are calculated here. All
$\rho_1=\rho(\boldsymbol{r}_1),\rho_2=\rho(\boldsymbol{r}_2)$ and
$\rho=\rho(\boldsymbol{r}_3)$ are \emph{pure} states). The points
under the horizontal axis indicate the violation of the triangular
inequality by triples of three qubit \emph{pure} states.}
\label{fig:Fig4}
\end{figure}

Denote
\begin{eqnarray*}
\boldsymbol{\Delta}'&=&\bD^2_{\hat
J}(\rho(\boldsymbol{r}_1),\rho(\boldsymbol{r}_2))+\bD^2_{\hat
J}(\rho(\boldsymbol{r}_1),\rho(\boldsymbol{r}_3))-\bD^2_{\hat
J}(\rho(\boldsymbol{r}_2),\rho(\boldsymbol{r}_3)).
\end{eqnarray*}
In FIG.~\ref{fig:Fig3} and FIG.~\ref{fig:Fig4}, we demonstrate the
fact that $\bD^2_{\hat J}$ is still not the true metric when
restricted to the set $\density{\complex^2}$ and the set of all pure
qubit states, respectively.

\subsection{The spectral density of mixture under the so-called quantum addition rule}

Assume that both $\rho_1$ and $\rho_2$ are i.i.d. random density
matrices for qubits, chosen from respective unitary orbits $\cO_\mu$
and $\cO_\nu$ with $\mu,\nu$ are fixed in the open interval
$(0,1/2)$. Denote
$\rho(\hat{\boldsymbol{r}})=\rho(\boldsymbol{r}_1)\boxplus_{1/2}\rho(\boldsymbol{r}_2)$.
Let $\hat r=\abs{\hat{\boldsymbol{r}}},r_1=\abs{\boldsymbol{r}_1}$
and $r_2=\abs{\boldsymbol{r}_2}$.

\begin{thrm}\label{th:qr-density}
The conditional probability density of $\hat r$ with respect to
fixed $r_1,r_2$ is given by
\begin{eqnarray}
q(\hat r|r_1,r_2)=\frac1{r_1r_2}\cdot\frac{2\hat
r}{\sqrt{(1+r^2_1)(1+r^2_2) -4\hat r^2}},
\end{eqnarray}
where $\hat r\in[r_-,r_+]$ with $r_-=\frac{\abs{r_1-r_2}}2$ and
$r_+=\frac{r_1+r_2}2$.
\end{thrm}

\begin{proof}
As already noticed above,
$$
\hat r=
\frac12\sqrt{r^2_1+r^2_2+2r_1r_2\cos\theta+r^2_1r^2_2\sin^2\theta}.
$$
That is, $\hat
r=\frac12\sqrt{(1+r^2_1)(1+r^2_2)-(1-r_1r_2\cos\theta)^2}$. Since
the probability density function of $\theta$ is given by
$\frac12\sin\theta\dif\theta$, it follows that
$\cos\theta,\theta\in[0,\pi]$ has the constant density $\frac12,
\cos\theta\in[-1,1]$. By change of variables, we finally get the
desired density $q(\widetilde r|r_1,r_2)$.
\end{proof}

From this theorem, we can infer the spectral density of the mixture
of two random density matrices for qubits under the quantum addition
rule:
$\rho_1\boxplus_{1/2}\rho_2=\frac12(\rho_1+\rho_2-\mathrm{i}[\rho_1,\rho_2])$.
\begin{cor}\label{cor:qlambda-density}
The spectral density of $\rho_1\boxplus_{1/2}\rho_2$, where $\rho_1$
and $\rho_2$ chosen uniformly from respective unitary orbits
$\cO_\mu$ and $\cO_\nu$ with $\mu,\nu$ are fixed in the open
interval $(0,1/2)$, is given by
\begin{eqnarray*}
q(\hat \lambda|\mu,\nu)
=\frac1{2\Pa{\frac12-\mu}\Pa{\frac12-\nu}}\frac{\abs{\hat
\lambda-\frac12}}{\sqrt{(2\mu^2-2\mu+1)(2\nu^2-2\nu+1) -
\Pa{2\hat\lambda-1}^2}},
\end{eqnarray*}
where $\hat \lambda\in[T_0,T_1]\cup[1-T_1,1-T_0]$. Note that $T_0$
and $T_1$ can be found in Proposition~\ref{prop:lmn}.
\end{cor}

\begin{proof}
The proof easily follows from Theorem~\ref{th:qr-density}.
\end{proof}

Here we give the figures to show how $p(\lambda|\mu,\nu)$ in
Proposition~\ref{prop:lmn} and $q(\hat \lambda|\mu,\nu)$ in
Corolloary~\ref{cor:qlambda-density} change with $\lambda$ and
$\hat\lambda$ when $\mu,\nu$ chosen, respectively, in
FIG.~\ref{fig:Fig5}; $p(r|r_1,r_2)$ in Theorem~\ref{th:rr1r2} and
$q(\hat r|r_1,r_2)$ in Theorem~\ref{th:qr-density} with $r$ and
$\hat r$, respectively, in FIG.~\ref{fig:Fig6}.

\begin{figure}[ht]
\centering {\begin{minipage}[b]{0.8\linewidth}
\includegraphics[width=1\textwidth]{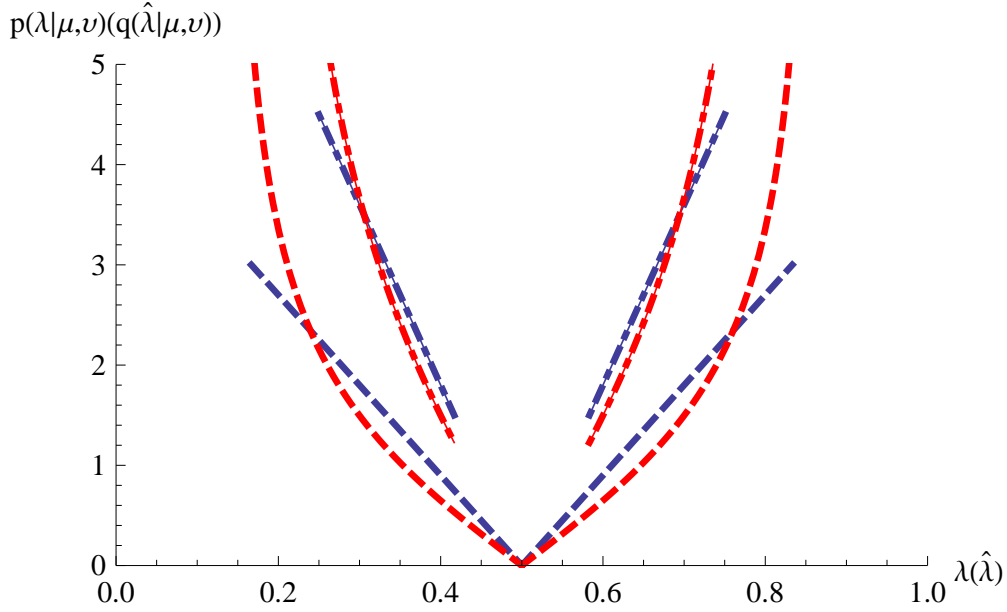}
\end{minipage}}
\caption{(Color online) The eigenvalue densities of two kinds of
mixtures of two qubit states. $p(\lambda|\mu,\nu)$ in
Proposition~\ref{prop:lmn} versus $\lambda$ and $q(\hat
\lambda|\mu,\nu)$ in Corollary~\ref{cor:qlambda-density} versus
$\hat\lambda$: $p(\lambda|\mu,\nu)$ is represented by the blue line
and $q(\hat\lambda|\mu,\nu)$ is represented by the red line. The
solid line ($\mu=\frac{1}{3}$,$\nu=\frac{1}{6}$), the dashed line
($\mu=\nu=\frac{1}{6}$) and the dot-dashed line
($\mu=\frac{1}{6}$,$\nu=\frac{1}{3}$), both the solid line and the
dot-dashed line are coincided.} \label{fig:Fig5}
\end{figure}

\begin{figure}[ht]
\centering  { \begin{minipage}[b]{0.8\linewidth}
\includegraphics[width=1\textwidth]{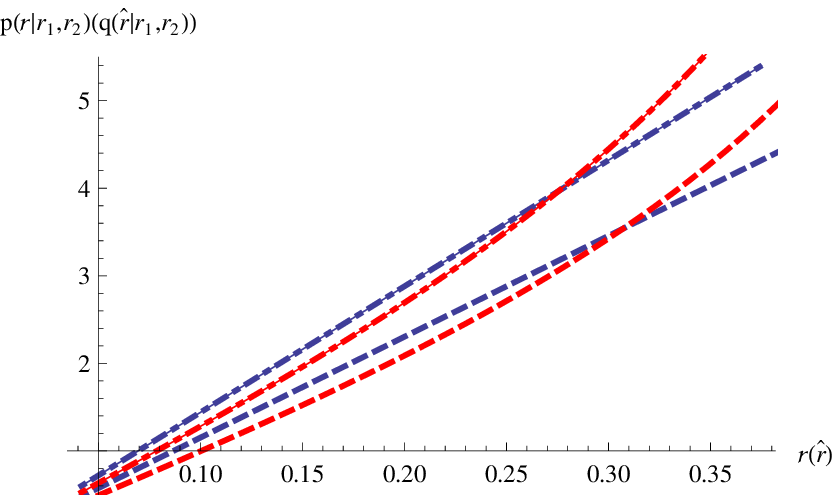}
\end{minipage}}
\caption{(Color online) The densities of the lengths of the Bloch
vectors for two kinds of mixtures of two qubit states.
$p(r|r_1,r_2)$ in Theorem~\ref{th:rr1r2} versus $r$ and $q(\hat
r|r_1,r_2)$ in Theorem~\ref{th:qr-density} versus $\hat r$:
$p(r|r_1,r_2)$ is represented by the blue line and $q(\hat
r|r_1,r_2)$ is represented by the red line. The solid line
($r_1=\frac13$,$r_2=\frac23$), the dashed line ($r_1=r_2=\frac23$)
and the dot-dashed line ($r_1=\frac23$,$r_2=\frac13$), both the
solid line and the dot-dashed line are coincided.} \label{fig:Fig6}
\end{figure}

\section{Further observations}\label{sect:app}

We have derived the spectral density of the equiprobable mixture of
two random density matrices for qubits, whose density function is
given by \eqref{eq:lambdamunu}. We can use this to calculate the
average entropy of the equiprobable mixture of two random density
matrices for qubits. Let $\mu,\nu\in[0,1/2]$. Choose
$\rho_1\in\cO_\mu$ and $\rho_2\in\cO_\nu$ according to the Haar
measure over the unitary group $\rS\rU(2)$. The average entropy of
the equiprobable mixture of two random density matrices is given by
\begin{eqnarray*}
\mathbb{E}_{(\rho_1,\rho_2)\in\cO_\mu\times
\cO_\nu}\rS\Pa{\frac{\rho_1+\rho_2}2}=
\iint\dif\mu_{\mathrm{Haar}}(U)\dif\mu_{\mathrm{Haar}}(V)
\rS\Pa{\frac{U\rho_1U^\dagger+V\rho_2V^\dagger}2}.
\end{eqnarray*}
By using the result in Theorem~\ref{th:rr1r2}, we get that
\begin{prop}
The average entropy of the equiprobable mixture of two random
density matrices chosen uniformly from orbits $\cO_\mu$ and
$\cO_\nu$, respectively, is given by the formula:
\begin{eqnarray}
\int^{r_+}_{r_-} \Phi(r)p(r|r_1,r_2)\dif r =
\frac2{r_1r_2}\int^{r_+}_{r_-} \Phi(r)r\dif r,
\end{eqnarray}
where $r_+:=\frac{r_1+r_2}2, r_-:=\frac{\abs{r_1-r_2}}2$ for
$r_1=\frac{1-\mu}2, r_2=\frac{1-\nu}2$, and $\Phi(x) =
-\frac{1+x}2\log_2\frac{1+x}2 - \frac{1-x}2\log_2\frac{1-x}2$.
\end{prop}

Denote
$$
\phi(x):=x^2(1+6\ln2)+(1-x)^2(1+2x)\ln(1-x)+(1+x)^2(1-2x)\ln(1+x).
$$
The average entropy is reduced to the following:
$$
\frac{\phi(r_+)-\phi(r_-)}{6\ln 2 \cdot r_1r_2}.
$$
Furthermore, we can calculate the average entropy of the
equiprobable mixture of two random states distributed according to
the Hilbert-Schmidt measure over $\density{\complex^2}$.
Specifically, using the distribution densities of $r_1$ and $r_2$,
we have
\begin{eqnarray*}
\mathbb{E}_{\rho_1,\rho_2} \rS\Pa{\frac{\rho_1+\rho_2}2} =
18\int^1_0\int^1_0\dif r_1\dif r_2 \Br{r_1r_2\int^{r_+}_{r_-}
\Phi(r)r\dif r}\approx 0.7631.
\end{eqnarray*}

Similarly, we can calculate the average entropy of the mixture of
two random density matrices for qubits under the quantum addition
rule: $\rho_1\boxplus_{1/2}\rho_2$, where $\rho_1\in\cO_\mu$ and
$\rho_2\in\cO_\nu$. The average entropy of the mixture
$\rho_1\boxplus_{1/2}\rho_2$ is given by
\begin{eqnarray*}
\mathbb{E}_{(\rho_1,\rho_2)\in\cO_\mu\times
\cO_\nu}\rS\Pa{\rho_1\boxplus_{1/2}\rho_2}=
\iint\dif\mu_{\mathrm{Haar}}(U)\dif\mu_{\mathrm{Haar}}(V)
\rS\Pa{U\rho_1U^\dagger\boxplus_{1/2}V\rho_2V^\dagger}.
\end{eqnarray*}

By using the result in Theorem~\ref{th:qr-density}, we have that
\begin{prop}
According to the quantum addition rule, the average entropy of the
mixture of two random density matrices chosen from orbits $\cO_\mu$
and $\cO_\nu$, respectively, is given by the formula:
\begin{eqnarray}
\int^{r_+}_{r_-} \Phi(\hat r)q(\hat r|r_1,r_2)\dif \hat r =
\frac2{r_1r_2}\int^{r_+}_{r_-} \frac{\Phi(\hat r)\hat r\dif \hat
r}{\sqrt{(1+r^2_1)(1+r^2_2) -4\hat r^2}},
\end{eqnarray}
where $r_+:=\frac{r_1+r_2}2, r_-:=\frac{\abs{r_1-r_2}}2$ for
$r_1=\frac{1-\mu}2, r_2=\frac{1-\nu}2$, and $\Phi(x) =
-\frac{1+x}2\log_2\frac{1+x}2 - \frac{1-x}2\log_2\frac{1-x}2$.
\end{prop}

Analogously, we can calculate the average entropy of the mixture
under the quantum addition rule of two random states distributed
according to the Hilbert-Schmidt measure over
$\density{\complex^2}$. Specifically,
\begin{eqnarray*}
&&\mathbb{E}_{\rho_1,\rho_2}\rS(\rho_1\boxplus_{1/2}\rho_2)\\
&&= 18\int^1_0\int^1_0\dif r_1\dif r_2 \Br{r_1r_2\int^{r_+}_{r_-}
\frac{\Phi(\hat r)\hat r\dif \hat
r}{\sqrt{(1+r^2_1)(1+r^2_2) -4\hat r^2}}}\\
&&\approx 0.7152.
\end{eqnarray*}

Recall that if a quantum system of Hilbert space dimension $mn$ is
in a random pure bipartite state, the average entropy of a subsystem
of dimension $m \leqslant n$ should be given by the simple and
elegant formula $H_{mn}-H_n-\frac{m-1}{2n}$, where
$H_k:=\sum^k_{j=1}1/j$ is the $k$-th harmonic number. This is
so-called Page's average entropy formula \cite{Page1993} which is
useful as a way of understanding the information in black hole
radiation. Using Page's formula, we see that
\begin{eqnarray*}
\frac13<
\mathbb{E}_{\rho_1,\rho_2}\rS(\rho_1\boxplus_{1/2}\rho_2)<\mathbb{E}_{\rho_1,\rho_2}
\rS\Pa{\frac{\rho_1+\rho_2}2}< 1.
\end{eqnarray*}
This indicates that our numerical calculations for
$\mathbb{E}_{\rho_1,\rho_2}\rS(\rho_1\boxplus_{1/2}\rho_2)\approx0.7152$
and
$\mathbb{E}_{\rho_1,\rho_2}\rS\Pa{\frac{\rho_1+\rho_2}2}\approx0.7631$
are compatible with the above inequality. It is also reasonably to
\emph{conjecture} that there is a chain of \emph{strict}
inequalities:
\begin{eqnarray}\label{eq:conjecture}
\mathbb{E}_{\rho_1,\rho_2}
\rS\Pa{\frac{\rho_1+\rho_2}2}<\mathbb{E}_{\rho_1,\rho_2,\rho_3}
\rS\Pa{\frac{\rho_1+\rho_2+\rho_3}3}<\cdots
<\mathbb{E}_{\rho_1,\ldots,\rho_n}
\rS\Pa{\frac{\rho_1+\cdots+\rho_n}n}< 1.
\end{eqnarray}
Indeed, denote $\Gamma = \sum^n_{j=1}\rho_j$ where each $\rho_j$ is
randomly chosen from Hilbert-Schmidt ensemble in a generic qubit
system. Thus
\begin{eqnarray*}
\frac{\sum^n_{j=1}\rho_j}n=\frac\Gamma n =
\frac1n\Pa{\sum^n_{j=1}\frac{\Gamma-\rho_j}{n-1}}.
\end{eqnarray*}
Then by the concavity of von Neumann entropy, we see that
\begin{eqnarray*}
\rS\Pa{\frac{\sum^n_{j=1}\rho_j}n}\geqslant
\frac1n\sum^n_{j=1}\rS\Pa{\frac{\Gamma-\rho_j}{n-1}}.
\end{eqnarray*}
Because $\rho_1,\ldots,\rho_n$ are independent and identically
distribution (i.i.d.), we have that, for each $j=1,\ldots,n$
\begin{eqnarray*}
\mathbb{E}_{\rho_1,\ldots,\rho_n}\rS\Pa{\frac{\Gamma-\rho_j}{n-1}}
=\cdots =
\mathbb{E}_{\rho_1,\ldots,\rho_{n-1}}\rS\Pa{\frac{\sum^{n-1}_{j=1}\rho_j}{n-1}}.
\end{eqnarray*}
Therefore, we obtain that
\begin{eqnarray*}
\mathbb{E}_{\rho_1,\ldots,\rho_n}\rS\Pa{\frac{\sum^n_{j=1}\rho_j}n}\geqslant\mathbb{E}_{\rho_1,\ldots,\rho_{n-1}}\rS\Pa{\frac{\sum^{n-1}_{j=1}\rho_j}{n-1}}.
\end{eqnarray*}
The strict inequality needs to be determined. Now we can use this
result and results obtained in \cite{Lin2018} to explain that the
quantum coherence \cite{Baumgratz2014} decreases statistically as
the mixing times $n$ increasing in the equiprobable mixture of $n$
qubits. Recall that the quantum coherence can be quantified by many
ways \cite{Baumgratz2014}. Here we take the coherence measure
defined via the relative entropy, i.e., the so-called \emph{relative
entropy of coherence}. The mathematical definition of the relative
entropy of coherence can be given as
$C_r(\rho):=\rS(\rho^D)-\rS(\rho)$, where $\rho^D$ is the diagonal
part of the quantum state $\rho$ with respect to a prior fixed
orthonormal basis. Denote $\overline{C}^{(n)}_r$ the average
coherence of the equiprobable mixture of $n$ i.i.d. random quantum
states from the Hilbert-Schmidt ensemble. By deriving the spectral
density of the mixture of 3 qubits and some analytical computations,
we show that $\overline{C}^{(3)}_r<\overline{C}^{(2)}_r$
\cite{Lin2018}. Numerical experiments further show that for any
integer $n>3$,
$$
\overline{C}^{(n)}_r<\cdots<\overline{C}^{(3)}_r<\overline{C}^{(2)}_r.
$$
Thus, in the qubit case, we find that the quantum coherence
monotonously decreases statistically as the mixing times $n$.
Moreover, we believe that the quantum coherence approaches zero when
$n\to\infty$.

Finally, we remark here that results in the present paper can be
used to compute exactly the average squared fidelity \cite{KZ2005}.
The fidelity between two qubit density matrices $\rho$ and $\sigma$
is defined by
$\rF(\rho,\sigma):=\Tr{\sqrt{\sqrt{\rho}\sigma\sqrt{\rho}}}$, then
\begin{eqnarray}
\mathbb{E}_{\rho,\sigma}\rF^2(\rho,\sigma) =
\frac12\Pa{1+\Pa{3\pi/16}^2}.
\end{eqnarray}
Indeed, it is easily seen that, the squared fidelity for the qubit
case is given by
\begin{eqnarray*}
\rF^2(\rho,\sigma) = 2\sqrt{\det(\rho)\det(\sigma)}+\Tr{\rho\sigma}.
\end{eqnarray*}
By using Bloch representations of $\rho=\rho(\boldsymbol{u})$ and
$\sigma=\rho(\boldsymbol{v})$, we have that
\begin{eqnarray*}
\rF^2(\rho(\boldsymbol{u}),\rho(\boldsymbol{v})) =
\frac12\Br{1+\Inner{\boldsymbol{u}}{\boldsymbol{v}} +
\sqrt{(1-u^2)(1-v^2)}},
\end{eqnarray*}
where $\abs{\boldsymbol{u}}=u\in[0,1]$ and
$\abs{\boldsymbol{v}}=v\in[0,1]$. In the following, we calculate the
average squared fidelity: Indeed, for $\boldsymbol{u}$ and
$\boldsymbol{v}$, denote by $\theta\in[0,\pi]$ the angle between
$\boldsymbol{u}$ and $\boldsymbol{v}$, their joint distribution
density is given by
$$
p(\boldsymbol{u},\boldsymbol{v})=p(u,v,\theta) = \frac92
u^2v^2\sin\theta,\quad \theta\in[0,\pi].
$$
Thus direct calculation gives rise to the desired result.
\begin{eqnarray*}
\mathbb{E}_{\rho,\sigma}\rF^2(\rho,\sigma) &=&
\mathbb{E}_{\boldsymbol{u},\boldsymbol{v}}
\rF^2(\rho(\boldsymbol{u}),\rho(\boldsymbol{v}))\\
&=&
\int\rF^2(\rho(\boldsymbol{u}),\rho(\boldsymbol{v}))p(\boldsymbol{u},\boldsymbol{v})[\dif\boldsymbol{u}][\dif\boldsymbol{v}]\\
&=&\frac12\Pa{1+\Pa{3\pi/16}^2}.
\end{eqnarray*}

\section{Concluding remarks}

In this paper, we work out the spectral densities of two kinds of
mixtures, i.e., the equiprobable mixture and the mixture under the
quantum addition rule, of two-level random density matrices chosen
uniformly from the Haar-distributed unitary orbits, respectively.
Before our work in the present paper, researchers always focus on
the eigenvalue statistics for individual quantum state ensemble, and
used frequently free probabilistic tools to make asymptotic analysis
to get much information about some statistical quantities. Although
the spectral analysis of superposition of random pure states were
performed recently \cite{ZPNC2011,CFF2013}, the topic about the
spectral densities for the mixtures of random density matrices from
two quantum state ensembles is rarely touched upon previously.
Moreover, our methods in the present paper can be further used to
derive the spectral densities of two kinds of mixtures:
$w\rho_1+(1-w)\rho_2$ and $\rho_1\boxplus_w\rho_2$ for qubits, where
$w\in[0,1]$. Finally, we leave some open questions here: (i) Can
Proposition~\ref{prop:CP-PSC} be generalized to a general qudit
system? (ii) We conjecture that
$\rS(\rho_1\boxplus_w\rho_2)\leqslant \rS\Pa{w\rho_1+(1-w)\rho_2}$
for any $w\in[0,1]$ and $\rho_1,\rho_2\in\density{\complex^d}$. We
believe that our contribution in exact spectral analysis of the
mixtures of random states will spur more new developments of
applying RMT in quantum information theory.

\subsubsection*{Acknowledgments}

LZ is truly grateful to the reviewer's comments and suggestions on
improving the manuscript. LZ acknowledges Huangjun Zhu for improving
the quality of our manuscript. LZ is supported by Natural Science
Foundation of Zhejiang Province (LY17A010027) and NSFC
(Nos.11301124, 61771174), and also supported by the
cross-disciplinary innovation team building project of Hangzhou
Dianzi University. Both JW and ZC are also supported by NSFC
(Nos.11401007,11571313).



\end{document}